\documentclass[english,12pt,draftcls,onecolumn]{IEEEtran}
\usepackage{caption}
\usepackage{epsfig}
\usepackage{amsthm,amssymb,graphicx,graphicx,multirow,amsmath,color,mathtools}
\usepackage{tikz}
\usepackage{pgfplots}
\usepackage{circuitikz}
\newcommand{\ignore}[1]{}
\usepackage{epstopdf}
\usetikzlibrary{circuits.logic.US,calc}
\usetikzlibrary{patterns}
\usepackage{verbatim}
\usepackage{psfrag}
\usepackage{bigints}

\def\E{\mathbb{E}}

\def\P{\mathbb{P}}

{}
  \newtheorem{thm}{Theorem}
  \newtheorem{theorem}{Theorem}
  
  \newtheorem{cor}[thm]{Corollary}

   {}

\usepackage{subcaption}
\usepackage{tikz}
\usepackage{tabularx}
\providecommand{\keywords}[1]{\textbf{\textit{Keywords ---}} #1}
\title{Ranging success probability of PPP distributed automotive radar in presence of generalized fading}
\author{\IEEEauthorblockN{ Sudharsan Parthasarathy$^{(a)}$, Rakshith Jagannath$^{(b)}$,}\\
\IEEEauthorblockA{
${}^{(a)}$Dept. of Electronics and Communication Engineering, National Institute of Technology  Tiruchirappalli, India\\
${}^{(b)}$School of Electrical and Electronic Engineering, Nanyang Technological University, Singapore\\
sudharsan@nitt.edu, rakshith.jagannath@ntu.edu.sg}
}
\begin{document}
\maketitle
\begin{abstract}
In automotive radar applications, multiple radars are used in all vehicles for improving the imaging quality. However this causes radar-to-radar interference from neighbouring vehicles, thus reducing the imaging quality. One metric to measure the imaging quality is ranging success probability. The ranging success probability is the probability that a multiple radar system successfully detects an object at a given range, under certain operating conditions.   In state-of-the-art literature,  closed form expressions for ranging success probability have been derived assuming no fading in desired signal component.   Similarly in literature,  though  distribution of fading in interferers is assumed to be arbitrary, closed form expression is derived only for no-fading assumption in interferers. As fading is always present in a wireless channel, we have derived ranging success probability assuming desired channel experiences the popular Rayleigh fading. And we have assumed generalized $\kappa$-$\mu$ shadowed fading for interfering channels that generalizes many popular fading models such as Rayleigh, Rician, Nakagami-$m$, $\kappa$-$\mu$ etc.  The interferers are assumed to be located on points drawn from a Poisson point process distribution. We have also studied how the relationship between shadowing component and number of clusters can affect the impact of LOS component on ranging success probability. 
\end{abstract}

\section{Introduction}
Automotive radar is the key component in successful development of driverless cars \cite{heathradar}.  With increase in deployment of driverless cars, interference from other vehicles will affect radar operations \cite{radarinterferencemitigation}. To quantify the impact of interference, a very popular metric that is used is ranging success probability, which is defined as the probability of reliably detecting a target given a certain set of operating conditions \cite{radar}. 

Applying stochastic geometry based modeling for automotive radar applications has been pioneered in \cite{radar}. To model the interference, we assume that the interfering vehicles are distributed in a 1-dimensional Poisson point process (PPP) as in \cite{radar}. PPP based modeling of automotive radar has gained a lot of attention recently in \cite{jeyaraj2017}, \cite{munari2018stochastic}. PPP based modeling of interference is already very popular in cellular networks \cite{andrews2011tractable}. Apart from modeling modern networks very well, PPP based modeling also helped in deriving closed form expression of coverage probability in cellular networks.

In \cite{radar}, fading is not considered in desired channel and interferers are assumed to experience arbitrary fading. Closed form expression for ranging success probability was derived assuming no fading for interfering channels.  Unlike cellular networks, here the desired signal  received back at the source is composed of the signal from the source to target and the signal back from target to source. Hence the desired channel power at the source is a product of the power received at target and power received back at source. If fading is assumed, the desired power will be product of two independent channel power random variables, making the analysis difficult. In this paper we assume the popular Rayleigh fading distribution for the desired channel. Hence the desired channel power is product of independent exponentially distributed random variables whose PDF can be expressed in terms of Bessel function. Hence, we have used Gamma-Laguerre integral approach to express the PDF as weighted sum of Gamma distribution, to derive the ranging success probability. To the best of our knowledge, such analysis of a ``product channel" has not been done in literature  for a PPP network. 

In this paper we assume the interferers experience generalized $\kappa$-$\mu$ shadowed fading \cite{paris2014statistical} and derive a closed form expression for ranging success probability.  The advantage is that $\kappa$-$\mu$ shadowed fading generalizes popular fading distributions such as Rayleigh, Rician, Nakagami etc.  Hence closed form expressions of ranging success probability  can be derived as special cases when interferers experience Rayleigh, Rician, Nakagami, $\kappa$-$\mu$ fading etc. also

Previously, $\kappa$-$\mu$ shadowed fading has been used extensively for analysis of cellular networks, as it generalizes popular fading models. In \cite{sudharsancoverage}, closed form expression for coverage probability of a cellular network was derived in the presence of $\kappa$-$\mu$ shadowed fading. In \cite{chun1},  \cite{trigui2017unified}  5G cellular networks were analysed in the presence of $\kappa$-$\mu$ shadowed fading. Similarly, communication networks were analysed in the presence of generalized fading for mm-wave \cite{kibilda2018performance} and device-to-device technologies \cite{chun2017stochastic}.

The PDF of $\kappa$-$\mu$ shadowed fading is in terms of Hypergeometric function. To make it amenable for analysis, PDF of $\kappa$-$\mu$ shadowed fading has been expressed in terms of density function of Gamma distribution in different ways. In \cite{sudharsancoverage}, by using definition of ${}_1F_1$ hypergeometric function, PDF of  $\kappa$-$\mu$ shadowed fading channel power was expressed as infinite sum of weighted Gamma density functions. In \cite{kumar2015approximate}, PDF of $\kappa$-$\mu$ shadowed fading was approximated by a single Gamma distribution using moment matching approach.  In \cite{lopez2017kappa} for integer parameters, PDF of $\kappa$-$\mu$ shadowed fading was expressed as finite sum of Gamma density functions. 

To the best of our knowledge, for the first time, in this paper we have studied how the relationship between shadowing component (m) and number of clusters ($\mu$) can affect the impact of LOS component ($\kappa$) on ranging success probability. We show that when there is full shadowing ($m$=0.5), with increase in LOS component $\kappa$ in interferer, $P_s$ increases. Similarly when there is no shadowing component ($m$=$\infty$), with increase in LOS component $\kappa$ in interferer, $P_s$ decreases.

As a practical application, our results can also be used to suggest when it would be best for city planners to increase longitudinal distance ($\delta_0$) for a given inter-lane distance or vice versa. We also show the importance of results derived in this paper assuming fading by comparing it with no-fading cases in desired and interfering channels. We have also observed how the value of distance to desired target ($R$) can influence the impact of change in desired path loss exponent ($\alpha_d$) on ranging success probability. Similarly, we have also observed how the value of longitudinal distance ($\delta_0$) can influence the impact of change in interferer path loss exponent ($\alpha_I$) on ranging success probability. The formulae derived for different fading scenarios will help system planners get an idea of SINR threshold required to achieve target success probability, for any of the popular fading that they observe in real life scenarios.

In Section II, system model is provided. In Section III, contributions of this paper and notations used are discussed. In Section IV, ranging success probability when desired channel experiences Rayleigh fading is derived. In Section V, ranging success probability when desired channel experiences no-fading is derived. In Section VI, results are discussed in great detail. In Section VII, conclusions and future work are provided. 

 \begin{center}
\begin{tabular}{ |c|c|c|c| } 
\hline
Desired\\channel &Interferer fading& Equation no.\\
\hline
\multirow{3}{4em}{Rayleigh-Rayleigh} & $\kappa$-$\mu$ shadowed & Exact \eqref{eqn:Ps_kappamushadowed} \\  \cline{2-3}
& $\kappa$-$\mu$ shadowed  & Approx. \eqref{eqn:Ps_kappamushadow_rayleigh_approximate} \\ \cline{2-3}
& Rician shadowed & \eqref{eqn:Ps_kappamushadow_rayleigh_approximate}  for $\mu$=1,$\kappa$=K    \\  \cline{2-3}
& $\kappa$-$\mu$ & \eqref{eqn:Ps_kappamu}  \\  \cline{2-3}
& Rician & \eqref{eqn:Ps_kappamu} for $\mu$=1,$\kappa$=K  \\  \cline{2-3}
&Nakagami-m & \eqref{eqn:Ps_nakagami_desiredrayleigh} \\ \cline{2-3}
&Rayleigh & \eqref{eqn:Ps_nakagami_desiredrayleigh}  for $\hat{m}$=1 \\ \cline{2-3}
&Rayleigh,  arbitrary L & \eqref{eqn:Ps_rayleigh_desiredrayleigh_L} \\ \cline{2-3}
\hline
\multirow{3}{4em}{No fading} & $\kappa$-$\mu$ shadowed & Exact \eqref{eqn:Ps_kappamushadowedfading} \\  \cline{2-3}
& $\kappa$-$\mu$ shadowed  & Approx. \eqref{eqn:Ps_kappamushadowed_approximate} \\ \cline{2-3}
& Rician shadowed  &  \eqref{eqn:Ps_kappamushadowed_approximate} for $\mu$=1,$\kappa$=K    \\ \cline{2-3}
& $\kappa$-$\mu$  & \eqref{eqn:Ps_kappamufading} \\ \cline{2-3}
& Rician & \eqref{eqn:Ps_kappamufading} for $\mu$=1,$\kappa$=K   \\ \cline{2-3}
& Nakagami-m & \eqref{eqn:Ps_nakagami} \\ \cline{2-3}
& Rayleigh & \eqref{eqn:Ps_rayleigh} \\ \cline{2-3}
\hline
\end{tabular}
\end{center}
Table I: List of Contributions

\begin{tabular}{ |p{1.3cm}|p{6cm}|  }
 \hline
 \hspace{-1mm}Notation & Description \\
 \hline
 S  & Signal Power  \\
\hline
 I  & Total interference power  \\
\hline
 $I_x$  & Interference power at distance x from source  \\
\hline
$\alpha_d$  & Path loss exponent of desired channel  \\
\hline
$\alpha_I$  & Path loss exponent of interferer  \\
\hline
R  & Distance between source and target  \\
\hline
L & Distance between lanes \\
\hline
$\delta_0$ & Minimum horizontal distance to interferer \\
\hline
$P_0$ & Power of transmission from source \\
\hline
$G_t$ & Transmit antenna gain \\
\hline
$\sigma_c$ & Radar cross section area of the target \\
\hline
$A_e$ & Effective area\\
\hline
$g_0$ & Fading channel power from source to target \\
\hline
$g_0^{\prime}$ & Fading channel power from target to source \\
\hline
$\overline{g_0}$ & Fading channel power from source to target \\
\hline
$\overline{g_0^{\prime}}$ & Fading channel power from target to source \\
\hline
$\gamma_0$ & $G_t A_e P_0/(4 \pi)$ \\
\hline
$\gamma_1$ & $\sigma_c/(4 \pi)$ \\
\hline
$c_0$ & $\gamma_0 \gamma_1$ \\
\hline
$\theta$ & Beamwidth of antenna \\
\hline
$P_s$ & Ranging success probability \\
\hline
N & Number of weights\\
\hline
\end{tabular}

\begin{tabular}{ |p{1.3cm}|p{6cm}|  }
 \hline
 \hspace{-1mm}Notation & Description \\
 \hline
$\sigma^2$ & AWGN power \\
\hline
$\rho_I$ & Density of PPP of interferers \\
\hline
$\mu$ & Number of clusters in $\kappa$-$\mu$ shadowed fading\\
\hline
$\kappa$ & LOS component in $\kappa$-$\mu$ shadowed fading \\
\hline
$m$ & Shadowing component in $\kappa$-$\mu$ shadowed fading \\
\hline
$g_x$ & Power of interference fading channel \\
\hline
$K$ & Rician shape parameter \\
\hline
$\hat{m}$ & Nakagami shape parameter \\
\hline
$\overline{g_p}$ &  $\overline{g_0}$ $\overline{g_0^{\prime}}$\\ 
\hline
\end{tabular}

Table II: List of Notations

\section{System Model}
\begin{figure}
\begin{tikzpicture}
\draw (0,0) -- (1,0) -- (1,1) -- (0,1) -- (0,0);
\draw(0.5,0.5) node{Source};
\draw (7,0) -- (8,0) -- (8,1) -- (7,1) -- (7,0);
\draw(7.5,0.5) node{Target};
\draw (1,0.5) -- (7,0.5);
\draw(4,0.8) node{R};
\draw (5,2) -- (7,2) -- (7,3) -- (5,3) -- (5,2);
\draw(6,2.5) node{Interferer 1};
\draw (1,0.5) -- (5,2.5);
\draw(3,1.3) node{x};
\draw (1,0.5) -- (1,2.5);
\draw (5,2.5) -- (1,2.5);
\draw(1.2,1.5) node{L};
\draw(2,2.7) node{$\delta_0$};
\draw (1,.5) -- (1.5,.8);
\draw (1,.5) -- (1.5,.2);
\draw(1.3,.5) node{$\theta$};
\end{tikzpicture}
\label{fig:system}
\caption{System Model}
\end{figure}
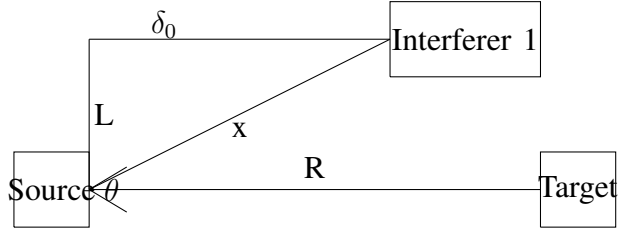

We consider the same system model (Fig. 1) as in \cite{radar} for a fair comparison.  The source is at a distance R from the target. Let $L$ be the distance of separation to the lane in which interfering vehicles move in opposite direction. Let $\delta_0$ be the minimum horizontal distance beyond which the vehicles coming in opposite lane act as interferers. The distance $\delta_0$ is related to the beamwidth of the antenna. 

\subsection{Signal Power}
 Let $P_0$ be the transmit radar power, $G_t$ be the transmit antenna gain, $A_e$ be the effective area, $\sigma_c$ be the radar cross section area of the target, $g_0$ be the channel fading power from source to target, $g_0^{\prime}$ be the channel fading power from target to source. The signal from source radar hits the target and returns to the source. So the received signal power at source is given as 
\begin{align}
S &=\frac{P_0 G_t g_0}{4 \pi R^{\alpha_d}} \frac{\sigma_c A_e g_0^{\prime}}{4 \pi R^{\alpha_d}} \nonumber \\  & =  \gamma_0 \gamma_1  g_0 g_0^{\prime} R^{-2 \alpha_d} \nonumber \\ &=  c_0 g_p R^{-2 \alpha_d} 
\label{eqn:signalpower}
\end{align}
where $g_p=g_0 g_0^{\prime},$ $\gamma_0=\frac{G_t A_e P_0}{4 \pi}$ =$G_t^2 P_0 (\frac{c}{4 \pi f})^2,$  $\gamma_1=\frac{\sigma_c}{4 \pi},$ $c_0=\gamma_0 \gamma_1 ,$ 
f is the operating frequency, c is the velocity of light, $\alpha_d$ is the desired channel's path loss exponent.  When $\alpha_d$=2, the standard radar equation, which follows the inverse square law  is obtained.

\subsection{Interference power}
The vertical distance $L$ between two lanes is related to beamwidth of antenna as \cite{radar} $L= \delta_0 tan(\frac{\theta}{2}).$ Interference from interferer at distance $||x||$, with fading power $g_x$ is given as  $I_x= \gamma_0  ||x||^{-\alpha_I} |g_x| .$ Total interference power is  
\begin{equation}
I= \sum\limits_{x \epsilon \Phi} I_x,
\label{eqn:interference}
\end{equation}
where $\Phi$ is the homogeneous Poisson point process of interferers with density $\rho_I$. 

The minimum distance of interferer is 
\begin{equation}
||x||=\sqrt{r^2+L^2},  r>\delta_0
\label{eqn:mindist}
\end{equation}
Ranging success probability  is  defined as 
\begin{equation}
P_s =  \P \left(\frac{S}{I+\sigma^2} >T \right)  
\label{eqn:Ps}
 \end{equation}
 where $\sigma^2$ is the additive white Gaussian noise power.
 \section{Contributions and Notations}

\vspace{5mm}
First we will derive the ranging success probability assuming the desired channels experience Rayleigh fading. Hence in \eqref{eqn:signalpower}, $g_0$, $g_0^{\prime}$ are independent and exponentially distributed of mean $\overline{g_0}$, $\overline{g_0^{\prime}}$. For this, the ranging success probability will be derived by expressing the ranging success probability in terms of Laplace transform of interference. This is done by expressing product of density function of $g_0$, $g_0^{\prime}$ which is a modified Bessel function of second kind, zero order in terms of weighted sum of exponential PDF using Gaussian Laguerre method.
In \cite{radar} no fading was considered in desired channel. Though arbitrary fading was considered in interferers, closed form results were derived only for special case of no-fading. Here we assume the interferers to experience generalized $\kappa$-$\mu$ shadowed fading and derive ranging success probability for many special cases of fading in interferers such as Rayleigh, Rician, Nakagami-$m$, $\kappa$-$\mu$ etc. In \cite{radar}, ranging success probability was derived for specific parameters such as path loss exponent of 2, $L$=0, $\delta_0$=0. This is because.  to derive  ranging success probability, without considering fading in desired channel, requires the CDF of interference. CDF of interference can be derived in closed form, only for these parameters. Whereas by considering fading, we have expressed PDF of desired signal power in terms of weighted sum of exponential PDF. This enables us to express ranging success probability in terms of  Laplace transform of interference alone. As Laplace transform of interference need not be inverted to derive the PDF/CDF of interference, ranging success probability can be derived for arbitrary parameters of $\alpha$, $L$, $\delta_0$. 

Next, ranging success probability is derived  when there is no fading in the desired channel as in \cite{radar} i.e. when $g_0=1$, $g_0^{\prime}=1$ in \eqref{eqn:signalpower}, but by assuming the fading in interference to experience $\kappa$-$\mu$ shadowed fading. In \cite{radar}, closed form expressions were given only when interfering channels experience no-fading. From the ranging success probability derived for $\kappa$-$\mu$ shadowed fading model in interferers, the same can be derived when interferers experience Rayleigh, Nakagami-$m$, Rician, $\kappa$-$\mu$ fading etc. as special cases. Also ranging success probability is derived when there is no-fading in interference as a special case, validating the expression derived in \cite{radar}.

Our contributions mentioned above, are listed in Table I to provide a quick overview to readers.  The list of notations that are used in this paper is given in Table II. In the next Section, we will derive  the ranging success probability when desired channel experiences Rayleigh fading.

\section{Ranging success probability when desired channel experiences Rayleigh fading}

In this Section, ranging success probability is derived when desired channel experiences Rayleigh fading, interfering channels experience generalized $\kappa$-$\mu$ shadowed fading channel. Also ranging success probability is derived when interferers experience other fading models such as Rayleigh, Rician, Nakagami-m, $\kappa$-$\mu$ etc. as special cases.

In \eqref{eqn:Ps}, ranging success probability is defined as 
\begin{align*}
P_s &= P(\frac{S}{I+\sigma^2}>T)  \\
&= P(S>T(I+\sigma^2))
\end{align*}
Substituting for signal power S from \eqref{eqn:signalpower},
\begin{align}
P_s &= P(c_0 g_p R^{-2 \alpha_d}>T(I+\sigma^2)) \\
&= P \left(g_p > \frac{T(I+\sigma^2) R^{2 \alpha_d}}{c_0} \right)
\label{eqn:Ps_gp}
\end{align}
So we have to derive the PDF and CDF of the product channel $g_p=g_0 g_0^{\prime}$.

As the fading of desired channel is Rayleigh, $g_0$ and $g_0^{\prime}$ are independent and exponentially distributed of mean $\overline{g_0}$, $\overline{g_0^{\prime}}$ respectively. Product of random variables $g_0$ and $g_0^{\prime}$ is $g_p$ whose mean is $\overline{g_p}= \overline{g_0} \overline{g_0^{\prime}} $.

So PDF of $g_p$ is given in \cite{exponential_thesis} as
\[f_{g_p} (x)=  \frac{2 K_0(2 \sqrt{ \frac{x}{  \overline{g_p } }} ) }{  \overline{g_p }}\]
where $K_0$ is modified Bessel function of second kind of zeroth order defined as
\[K_0(z) =\frac{1}{2} \int\limits_{0}^{\infty}  \frac{e^{-z(y+1/y)/2}}{y} d y \]
\[f_{g_p} (x)=  \frac{1}{  \overline{g_p }}  \int\limits_{0}^{\infty}  \frac{e^{-\sqrt{ \frac{x}{  \overline{g_p } }} (y+1/y)}}{y} d y  \]
Substituting $t=\sqrt{\frac{x}{ \overline{g_p } }} y $,
\begin{equation}
f_{g_p} (x)=  \frac{1}{  \overline{g_p }}  \int\limits_{0}^{\infty}  \frac{e^{-(t+\frac{x}{ \overline{g_p}t})}}{t} d t  
\label{eqn:fgp}
\end{equation}

In \cite{atapatu}, it is given that if 
\begin{equation}
f_{\gamma}(x) =\frac{\lambda^m x^{m-1}}{ \Gamma(m) \Gamma(k)} \int_0^{\infty} e^{-t} g(t) 
\label{eqn:fgamma}
\end{equation}
where 
\begin{equation}
g(t)=t^{\alpha-1} e^{-\frac{ \lambda x}{t}},
\label{eqn:fgamma_gt}
\end{equation}
$\lambda=\frac{km}{\overline{\gamma}}$, $\alpha=k-m$, then $I= \int_0^{\infty} e^{-t} g(t)$ can be represented as $I \approx \sum\limits_{i=1}^N w_i g(t_i)$. The weights $w_i$ and abscissas $t_i$ are obtained using Gaussian-Laguerre integration \cite{handbook}. Hence $f_{\gamma}(x)$ can be represented as weighted sum of PDF of Gamma distributions given as $\sum\limits_{i=1}^N w_i f_i(x)$ where $f_i(x)=\frac{ \Psi_i^{\beta_i} x^{\beta_i-1} e^{-\Psi_i x}}{ \Gamma(\beta_i)}$. The parameters are $\beta_i=m$, $\Psi_i=\frac{\lambda}{t_i}$.

Comparing \eqref{eqn:fgamma} with \eqref{eqn:fgp}, $\lambda=\frac{1}{\overline{g_p}}$, $m$=1,
\begin{equation}
g(t)=e^{-\frac{x}{\overline{g_p}t}}.
\label{eqn:fgp_gt}
\end{equation}
Comparing \eqref{eqn:fgamma_gt} with \eqref{eqn:fgp_gt}, $\alpha=0$, $k=1$, we get
\begin{equation}
    f_{g_p}(x) =\sum\limits_{i=1}^N \frac{w_i e^{-\frac{x}{\overline{g_p} t_i  }}}{ \overline{g_p} t_i  } 
    \label{eqn:fgp_gamma}
    \end{equation}
    
    So $f_{g_p}(x)$ is a weighted sum of exponentials of mean $\overline{g_p} t_i  $. The weights $w_i$ and abscissas $t_i$ are computed numerically in Matlab or Mathematica using standard functions for Gaussian-Laugerre method.
As weights $w_i$ sum to 1, CCDF is 
\begin{equation}
    \overline{F}_{g_p}(x) =\sum\limits_{i=1}^N w_i e^{-\frac{x}{\overline{g_p} t_i  }}
    \label{eqn:Fgp_gamma}
    \end{equation}
\subsection{$\kappa-\mu$ shadowed fading interferer channels}
In this Section, ranging success probability is derived when desired channel experiences Rayleigh fading and interferers experience $\kappa$-$\mu$ shadowed fading. First, an accurate expression is derived in Theorem \ref{theorem:kappamushadowed_rayleigh} and then a simpler approximation is derived in Corollary \ref{cor:kappamushadowed_rayleigh_approximate}.
\begin{theorem}
For arbitrary $\alpha_d$, $\alpha_I$, $\delta_0$, when desired channel experiences Rayleigh fading and interferers experience $\kappa$-$\mu$ shadowed fading, ranging success probability $P_s$ is
\begin{equation}
\scriptstyle
\sum\limits_{i=1}^N   \frac{w_i e^{-\rho_I\sum\limits_{l=0}^{\infty} w_{Il}\delta_0 (-1 + {}_2F_1(-\frac{1}{\alpha_I}, l+\mu, \frac{\alpha_I-1}{\alpha_I}, -\frac{T R^{2 \alpha_d} \delta_0^{-\alpha_I} \gamma_0}{c_0 t_i \overline{g_p} \mu (1+\kappa) })) } }{  e^{\frac{T \sigma^2 R^{2 \alpha_d}}{c_0 t_i \overline{g_p}}} }
\label{eqn:Ps_kappamushadowed}
\end{equation}
where $w_{Il}=\frac{ \Gamma(m+l) (\mu \kappa)^l m^m }{\Gamma(m) l! (\mu \kappa+m)^{l+m} }$, $w_i$  and $t_i$ are weights and abscissas found from Gaussian-Laugerre method.
\label{theorem:kappamushadowed_rayleigh}
\end{theorem}
\begin{proof}
Ranging success probability from \eqref{eqn:Ps_gp} is
\[P_s=  \overline{F}_{g_p}(\frac{T(I+\sigma^2) R^{2 \alpha_d}}{c_0})\]
Substituting CCDF from \eqref{eqn:Fgp_gamma}, 
\[P_s=E_I( \sum\limits_{i=1}^N w_i e^{-\frac{T(I+\sigma^2) R^{2 \alpha_d}}{c_0\overline{g_p} t_i  }} ) \]

\begin{equation}
P_s= \sum\limits_{i=1}^N w_i L_I(\frac{T R^{2 \alpha_d}}{c_0 t_i \overline{g_p}} )e^{-\frac{T \sigma^2 R^{2 \alpha_d}}{c_0\overline{g_p} t_i  }}  
\label{eqn:Ps_LI}
\end{equation}
where $L_I(s)$=$E_I(e^{-s I})$. So to derive the ranging success probability, $L_I(s)$ will be derived next.

Total interference power from \eqref{eqn:interference} is
\[ I= \sum\limits_{x \in \Phi} \gamma_0 g_x ||x||^{-\alpha_I} \]
Let $g_x$ be the power of the interferer $\kappa$-$\mu$ shadowed fading channel. Let $g_x$ be of unit mean power, without loss of generality.  Its probability density function is given as \cite{paris2014statistical}
\begin{equation}
f_{g_x}(x)= \frac{\mu^{\mu} m^m (1+\kappa)^{\mu}x^{\mu-1}  {}_1F_1(m;\mu;\frac{\mu^2 \kappa (1+\kappa) x}{\mu \kappa+m})}{e^{\mu (1+\kappa) x } \Gamma(\mu)  (\mu \kappa +m)^m} 
\label{eqn:fgx_kappamushadowed}
\end{equation}
where $_1F_1(a;b;z)=\sum\limits_{l=0}^{\infty} \frac{(a)_l z^l}{(b)_l l!}$, $(a)_l=\frac{\Gamma(a+l)}{\Gamma(a)}$. Using above definitions, and from \cite{sudharsancoverage}, we will represent the  PDF of power of $\kappa$-$\mu$ shadowed fading channel as weighted sum of PDF of Gamma distribution of shape and scale parameters $(l+\mu,\frac{1}{c})$.
So  \[f_{g_x}(x)= \sum\limits_{l=0}^{\infty} w_{Il}  \frac{e^{-c x} x^{l+\mu-1} c^{l+\mu} }{\Gamma(l+\mu)}\] where $c=\mu(1+\kappa),w_{Il}=\frac{\Gamma(l+\mu) (m)_l (\frac{\mu \kappa}{\mu \kappa+m})^l (\frac{m}{m+\mu \kappa})^m}{\Gamma(\mu) l! (\mu)_l}$. Using the above discussion, we will derive Laplace transform of interference.

\begin{align*}
    L_I(s) &= \E(e^{-s I})\\
    &= \E(e^{-s \sum\limits_{x \epsilon \Phi} I_x})\\
    &= \E_{I_x, \Phi}(\prod_{x \epsilon \Phi} e^{-s  I_x})\\
    &= \E_{g_x, \Phi}(\prod_{x \epsilon \Phi} e^{-s  \gamma_0 g_x ||x||^{-\alpha_I}})
\end{align*}
Using probability generating functional \cite{gantinow}, 
\[ \E_{\Phi}(\prod_{x \epsilon \Phi} f(x))= \exp \left(-\lambda \int\limits_R (1-f(x)) d x  \right)  \]
Hence
\begin{equation}
L_I(s) =\exp \left(-\rho_I \int\limits_{\delta_0}^{\infty} (1-\E_g(e^{-s \gamma_0 g x^{-\alpha_I}})) d r  \right)
\label{eqn:LI}
\end{equation}
where $x=\sqrt{r^2+L^2}$ from \eqref{eqn:mindist}.

For Gamma distributed random variable X of shape and scale parameters (m,n), 
\begin{equation}
\E(e^{-Xs})=\frac{1}{(1+sn)^m}.
\label{eqn:gamma_laplace}
\end{equation}
So for a $\kappa-\mu$ shadowed random variable whose probability density function can be represented as weighted sum of Gamma PDF of parameters $(l+\mu,\frac{1}{c})$, using \eqref{eqn:gamma_laplace}
\begin{equation}
 \E(e^{-sg})=\sum\limits_{l=0}^{\infty} \frac{w_{Il}}{(1+s/c)^{l+\mu}} 
 \label{eqn:kappamu_laplace}
 \end{equation}
As the weights sum to 1, substituting \eqref{eqn:kappamu_laplace} in \eqref{eqn:LI},
\begin{equation}
L_I(s) =e^{-\rho_I\sum\limits_{l=0}^{\infty} w_{Il} \int\limits_{\delta_0}^{\infty} (1-\frac{1}{(1+\frac{s \gamma_0}{c} (r^2+L^2)^{-\frac{\alpha_I}{2}})^{l+\mu}}) d r}
\label{eqn:LI_kappamushadowed_0}
\end{equation}
Initially we assume inter-lane distance $L$ to be much smaller than the longitudinal distance $r$ \cite{radar}. Later, we derive ranging success probability for arbitrary L too. 

Assuming $r>>L$, using the identity \[\scriptstyle \int\limits_{\delta_0}^{\infty} (1-\frac{1}{(1+s r^{-\alpha_I})^m}) d r =\delta_0 (-1 + {}_2F_1(-\frac{1}{\alpha_I}, m, 1-\frac{1}{\alpha_I}, -\frac{s}{\delta_0^{\alpha_I}}))\] in \eqref{eqn:LI_kappamushadowed_0}, we get 

\begin{equation}
L_I(s) =e^{-\rho_I\sum\limits_{l=0}^{\infty} w_{Il}\delta_0 (-1 + {}_2F_1(-\frac{1}{\alpha_I}, l+\mu, 1-\frac{1}{\alpha_I}, -\frac{s \gamma_0}{c d_0^{\alpha_I}})) }
\label{eqn:LI_kappamushadowed}
\end{equation}
Substituting \eqref{eqn:LI_kappamushadowed} in \eqref{eqn:Ps_LI} and $c=\mu(1+\kappa)$, $P_s$ in \eqref{eqn:Ps_kappamushadowed} is derived.
\end{proof}
In \cite{atapatu} the number of weights $N$ required to express PDF of different fading channels in terms of weighted sum of Gamma PDFs have been discussed. Similarly,  to express $\kappa$-$\mu$ shadowed fading PDF in terms of weighted sum of Gamma PDF, the number of weights ($w_{Il}$) required depends on  the parameters and is typically not very large \cite{sudharsanevm}. 

Instead of the accurate ranging success probability expression in \eqref{eqn:Ps_kappamushadowed}, a much simplified expression can be derived by approximating $\kappa$-$\mu$ shadowed random variable using a single Gamma random variable as in \cite{kumar2015approximate}.

\begin{cor}
For arbitrary $\alpha_d$, $\alpha_I$, $\delta_0$, when desired channel experiences Rayleigh fading and interferers experience $\kappa$-$\mu$ shadowed fading, approximate ranging success probability is
\begin{equation}
    P_s=\sum\limits_{i=1}^N \frac{w_i e^{-\rho_I \delta_0(-1+{}_2F_1(-\frac{1}{\alpha_I}, k, 1-\frac{1}{\alpha_I}, -\frac{T R^{2 \alpha_d} \gamma_0 \theta}{ c_0 \overline{g_p} t_i \delta_0^{\alpha_I}}  ))} }{e^{\frac{T \sigma^2 R^{2 \alpha_d}}{ c_0 \overline{g_p} t_i}}}
    \label{eqn:Ps_kappamushadow_rayleigh_approximate}
\end{equation}
where $(k,\theta)=( \frac{m  \mu (1+\kappa)^2}{ m+ \mu \kappa^2 
+2 m \kappa}, \frac{ m+ \mu \kappa^2 
+2 m \kappa}{m  \mu (1+\kappa)^2} )$.
\label{cor:kappamushadowed_rayleigh_approximate}
\end{cor}
\begin{proof}
The PDF of power of $\kappa-\mu$ shadowed fading interferer channel $g_x$ of unit mean, in \eqref{eqn:fgx_kappamushadowed} can be approximated by Gamma random variable of parameters $(k,\theta)=( \frac{m  \mu (1+\kappa)^2}{ m+ \mu \kappa^2 
+2 m \kappa}, \frac{ m+ \mu \kappa^2 
+2 m \kappa}{m  \mu (1+\kappa)^2} )$ \cite{kumar2015approximate}. Hence in \eqref{eqn:LI},
\[\E(e^{-s \gamma_0 g x^{-\alpha}}) = \frac{1}{(1+s \gamma_0 x^{-\alpha} \theta )^k} \]
where $x=\sqrt{r^2+L^2}$. 
\begin{equation}
    L_I(s)=e^{-\rho_I \int\limits_{\delta_0}^{\infty} (1-\frac{1}{(1+s \gamma_0 \theta (r^2+L^2)^{-\alpha_I/2})^k}) d r}
\label{eqn:LI_approx}
\end{equation}
Hence for $L$=0, following same steps as in Theorem \ref{theorem:kappamushadowed_rayleigh}, 
\begin{equation}
    L_I(s)=e^{-\rho_I \delta_0(-1+{}_2F_1(-\frac{1}{\alpha_I}, k, 1-\frac{1}{\alpha_I}, -\frac{s \gamma_0 \theta}{\delta_0^{\alpha}}  ))}
    \label{eqn:LI_kappamushadowed_rayleigh_approximate}
\end{equation}
Substituting \eqref{eqn:LI_kappamushadowed_rayleigh_approximate} in \eqref{eqn:Ps_LI}, $P_s$ in \eqref{eqn:Ps_kappamushadow_rayleigh_approximate} is derived. 
\end{proof}

 In \cite{kumar2015approximate} it has been shown that approximation is tight for different parameters of $\kappa$, $\mu$, $m$. We also show later through simulations that the ranging success probability in \eqref{eqn:Ps_kappamushadow_rayleigh_approximate} matches well with simulation. 

Next we will derive the ranging success probability when interferers experience different fading models such as Rayleigh, Rician, Nakagami-m etc. as special cases.
We derive these from the approximate ranging success probability in \eqref{eqn:Ps_kappamushadow_rayleigh_approximate}. By following same steps, special cases can be derived from accurate ranging success probability in \eqref{eqn:Ps_kappamushadowed} too.

For Nakagami-m, Rayleigh fading, the ranging success probability  derived  in Corollary \ref{cor:nakagami}-\ref{cor:rayleigh_L} from approximate $P_s$ in \eqref{eqn:Ps_kappamushadow_rayleigh_approximate} is same as deriving from accurate $P_s$ in \eqref{eqn:Ps_kappamushadowed}.  
\subsection{Rician shadowed fading}
For Rician shadowed fading, ranging success probability is derived from \eqref{eqn:Ps_kappamushadow_rayleigh_approximate}, for $\mu$=1, $\kappa$=K.

\subsection{$\kappa$-$\mu$ faded interferer channels}
In this Section we shall derive ranging success probability when interferers experience $\kappa$-$\mu$ fading, which is a special case of $\kappa$-$\mu$ shadowed fading for $m \rightarrow \infty$.
\begin{cor}
When desired channel experiences Rayleigh fading and interferers experience  $\kappa$-$\mu$ fading, ranging success probability is 
\begin{equation}
P_s=\sum\limits_{i=1}^N \frac{w_i e^{-\rho_I \delta_0(-1+{}_2F_1(-\frac{1}{\alpha_I}, k, 1-\frac{1}{\alpha_I}, -\frac{T R^{2 \alpha_d} \gamma_0 \theta}{ c_0 \overline{g_p} t_i \delta_0^{\alpha_I}}  ))} }{e^{\frac{T \sigma^2 R^{2 \alpha_d}}{ c_0 \overline{g_p} t_i}}}
\label{eqn:Ps_kappamu}
\end{equation}
where $(k,\theta)=( \frac{  \mu (1+\kappa)^2}{ 1 
+2  \kappa}, \frac{ 1 
+2  \kappa}{  \mu (1+\kappa)^2} )$.
\label{cor:kappamu}
\end{cor}
\begin{proof}
In \eqref{eqn:Ps_kappamushadow_rayleigh_approximate},  parameter $ k = \frac{m \mu (1+\kappa)^2}{ m (1 + \frac{\mu \kappa^2}{m}+2   \kappa)} .$
Allowing $m \rightarrow \infty$, $k = \frac{\mu (1+\kappa)^2}{  (1 +2   \kappa)}.$
Similarly $\theta$ is derived as $\theta=\frac{1}{k}$.
\end{proof}
\subsection{Rician  faded interferer channels}
For Rician fading, ranging success probability is derived from  \eqref{eqn:Ps_kappamu} by substituting $\mu$=1, $\kappa$=K.

\subsection{Nakagami-m  faded interferer channels}
In this Section we shall derive ranging success probability when interferers experience Nakagami-m fading of shape parameter $\hat{m}$, which is a special case of $\kappa$-$\mu$ fading for $\mu=\hat{m}$, $\kappa=0$.
\begin{cor}
When desired channel experiences Rayleigh fading and interferers experience  Nakagami-m fading, ranging success probability is 
\begin{equation}
P_s=\sum\limits_{i=1}^N   \frac{w_i e^{-\rho_I \delta_0 (-1 + {}_2F_1(-\frac{1}{\alpha_I}, \hat{m}, 1-\frac{1}{\alpha_I}, -\frac{T R^{2 \alpha_d} \gamma_0}{c_0 t_i \overline{g_p}  \hat{m} \delta_0^{\alpha_I}})) } }{  e^{\frac{T \sigma^2 R^{2 \alpha_d}}{c_0 t_i \overline{g_p}}} }
\label{eqn:Ps_nakagami_desiredrayleigh}
\end{equation}
\label{cor:nakagami}
\label{cor:rayleighdesired_nakagami}
\end{cor}
\begin{proof}
Nakagami-m fading of shape parameter $(\hat{m})$, scale parameter ($\frac{1}{\hat{m}}$) is a special case of $\kappa$-$\mu$ fading for $\mu=\hat{m}$, $\kappa=0$.  Substituting in \eqref{eqn:Ps_kappamu}, the ranging success probability in \eqref{eqn:Ps_nakagami_desiredrayleigh} is derived. 
\end{proof}
\subsection{Rayleigh fading}
Ranging success probability when interferers experience Rayleigh fading  is derived from \eqref{eqn:Ps_nakagami_desiredrayleigh} for $\hat{m}$=1.  For arbitrary inter-lane distance $L$, ranging success probability is derived next when interferers experience Rayleigh fading, for $\alpha_I=2$. 
\subsubsection{Ranging success probability for arbitrary L}
\begin{cor}
For arbitrary inter-lane distance L, $\alpha_d$, $\delta_0$, when desired channel experiences Rayleigh fading and interferers experience  Rayleigh fading, ranging success probability  $P_s$ is 
\begin{equation}
\sum\limits_{i=1}^N w_i e^{-\frac{T \sigma^2 R^{2 \alpha_d}}{ c_0 \overline{g_p} t_i}}   e^{-\rho_I \frac{\frac{ T R^{2 \alpha_d}\gamma_0}{c_0 t_i \overline{g_p}} ArcCot[\frac{\delta_0}{\sqrt{L^2+\frac{ T R^{2 \alpha_d}\gamma_0}{c_0 t_i \overline{g_p}} }}]}{\sqrt{L^2+\frac{ T R^{2 \alpha_d}\gamma_0}{c_0 t_i \overline{g_p}} }}  }
\label{eqn:Ps_rayleigh_desiredrayleigh_L}
\end{equation}
\label{cor:rayleigh_L}
\end{cor}
\begin{proof}
To derive ranging success probability, $L_I(s)$ is first computed. From \eqref{eqn:LI}
\begin{equation*}
L_I(s) =\exp \left(-\rho_I \int\limits_{\delta_0}^{\infty} (1-\E_g(e^{-s \gamma_0 g (r^2+L^2)^{-\alpha/2}})) d r  \right)
\end{equation*}
As the interferers are Rayleigh distributed of unit mean, and $\alpha_I$=2
\begin{align}
L_I(s) &=\exp \left(-\rho_I \int\limits_{\delta_0}^{\infty} (1-\frac{1}{1+s \gamma_0 (r^2+L^2)^{-1}}) d r  \right)  \nonumber \\
 &= \exp \left(-\rho_I \frac{s \gamma_0 ArcCot(\frac{\delta_0}{\sqrt{L^2+s \gamma_0}})}{\sqrt{L^2+s \gamma_0}}  \right) 
 \label{eqn:LI_rayleigh}
\end{align}
Substituting \eqref{eqn:LI_rayleigh} in \eqref{eqn:Ps_LI}, for $\alpha_I$=2, $P_s$ in \eqref{eqn:Ps_rayleigh_desiredrayleigh_L} is derived.
\end{proof}

\section{Ranging success probability when desired channel experiences no-fading }
In this Section, we derive ranging success probability when the desired channel experiences no-fading and interferers experience $\kappa$-$\mu$ shadowed fading. Using this result, ranging success probability is derived for special cases such as Rayleigh, Nakagami, Rician, $\kappa$-$\mu$ fading etc. This is a significant improvement over \cite{radar} in which closed form results were derived only for no-fading in interferers.  The difference from analysis in above Section is that, since the desired channel experiences no-fading, to derive the ranging success probability, Laplace transform of interference is not sufficient. Instead, CDF of interference will be required which restricts the parameters to be $\alpha_I=2$, $\delta_0=L=0$, as closed form expression for CDF of interference is not known for arbitrary parameters. 
\subsection{$\kappa-\mu$ shadowed fading interferer channels}
Similar to discussion in previous Section, first we derive accurate ranging success probability and then a simpler approximation is derived in Corollary \ref{cor:Ps_kappamushadowed_approximate}.
\begin{theorem}
When interferers experience $\kappa$-$\mu$ shadowed fading and desired channel experiences no-fading, for arbitrary $\alpha_d$, ranging success probability is
\begin{equation}
P_s=  Erfc \left(\frac{\rho_I \sqrt{\pi \gamma_0}  \sum\limits_{l=0}^{\infty} w_{Il} \frac{ \Gamma(0.5+l+\mu)}{\Gamma(l+\mu)} }{2 \sqrt{\mu(1+\kappa) (\frac{c_0 R^{-2 \alpha_d}}{T} -\sigma^2)}} \right) 
\label{eqn:Ps_kappamushadowedfading}
\end{equation}
$w_{Il}=\frac{ \Gamma(m+l) (\mu \kappa)^l m^m}{ (\mu \kappa +m)^{m+l} \Gamma(m) l! }$.
\label{thm:kappamushadowinterferer}
\end{theorem}
\begin{proof}

For the special case of $\alpha_I$=2, $L$=0, $\delta_0=0$ as considered in \cite{radar}, from \eqref{eqn:LI_kappamushadowed_0}
\begin{equation*}
L_I(s) =e^{-\rho_I\sum\limits_{l=0}^{\infty} w_{Il} \int\limits_{0}^{\infty} (\frac{(1+\frac{s \gamma_0}{c}r^{-2})^{l+\mu}-1}{(1+\frac{s \gamma_0}{c}r^{-2})^{l+\mu}}) d r}
\end{equation*}
Using the identity 
\begin{equation}
\int_0^{\infty} (1-\frac{1}{(1+s r^{-2})^m}) d r = \sqrt{\pi s} \frac{\Gamma(0.5+m)}{\Gamma(m)},
\label{eqn:identity0}
\end{equation}
\begin{equation*}
L_I(s) =\exp \left(-\rho_I \sqrt{s \gamma_0 \pi/c} \sum\limits_{l=0}^{\infty} w_{Il} \frac{ \Gamma(0.5+l+\mu)}{\Gamma(l+\mu)} \right)
\end{equation*}
Use Levy's distribution [(3.23) in \cite{gantinow}], i.e. if $L_I(s)=exp(-\sqrt{2as})$ then $F(x)=Erfc(\sqrt{\frac{a}{2x}})$. Hence , the CDF of interference is
\begin{equation}
 F_I(x)= Erfc \left(\frac{\rho_I \sqrt{\pi \gamma_0}  \sum\limits_{l=0}^{\infty} w_{Il} \frac{ \Gamma(0.5+l+\mu)}{\Gamma(l+\mu)} }{2 \sqrt{cx}} \right)
 \label{eqn:FI}
 \end{equation}
 From \eqref{eqn:Ps},
 \begin{equation}
     P_s=F_I \left(\frac{S}{T}-\sigma^2 \right)
     \label{eqn:Ps_FI}
 \end{equation}
 Therefore ranging success probability in \eqref{eqn:Ps_kappamushadowedfading} is obtained by combining \eqref{eqn:FI}, \eqref{eqn:signalpower}, \eqref{eqn:Ps_FI}, for $g_p$=1.   
\end{proof}
\begin{cor}
When interferers experience $\kappa$-$\mu$ shadowed fading and desired channel experiences no-fading, approximate ranging success probability is
\begin{equation}
P_s=  Erfc \left(\frac{\rho_I \sqrt{\pi \gamma_0 \theta} \Gamma(0.5+k)}{2 \Gamma(k) \sqrt{\frac{c_0 R^{-2 \alpha_d}}{T}-\sigma^2} } \right)
\label{eqn:Ps_kappamushadowed_approximate}
\end{equation}
where $k$=$\frac{m \mu (1+\kappa)^2}{m+\mu \kappa^2 +2 m \kappa}$ , $\theta=\frac{1}{k}$.
\label{cor:Ps_kappamushadowed_approximate}
\end{cor}
\begin{proof}
For the special case of $\alpha_I$=2, $L$=0, $\delta_0=0$ as considered in \cite{radar}, substituting \eqref{eqn:identity0} in \eqref{eqn:LI_approx},
\[L_I(s)=e^{-\rho_I \sqrt{\pi s \gamma_0 \theta} \frac{\Gamma(0.5+k)}{\Gamma(k)}} \]
Using Levy's distribution as in Theorem \ref{thm:kappamushadowinterferer},
\begin{equation}
    F_I(x)=Erfc \left(\frac{\rho_I \sqrt{\pi \gamma_0 \theta} \Gamma(0.5+k)}{2 \sqrt{x} \Gamma(k)} \right)
    \label{eqn:FI_approximate}
\end{equation}
Therefore ranging success probability in \eqref{eqn:Ps_kappamushadowed_approximate} is obtained  by combining \eqref{eqn:FI_approximate}, \eqref{eqn:signalpower}, \eqref{eqn:Ps_FI}, for $g_p$=1.   
\end{proof}

\subsection{Rician shadowed fading interferer channels}
Ranging success probability is  obtained from \eqref{eqn:Ps_kappamushadowed_approximate} by substituting $\mu$=1, $\kappa$=K, where $K$ is the Rician shape parameter. 

Next, ranging success probability from \eqref{eqn:Ps_kappamushadowedfading} is derived when interferers experience $\kappa$-$\mu$ fading. 

\subsection{$\kappa$-$\mu$ faded interferer channels}
\begin{cor}
When interferers experience $\kappa$-$\mu$ fading, ranging success probability is
\begin{equation}
P_s=  Erfc \left(\frac{\rho_I \sqrt{\pi \gamma_0 \theta} \Gamma(0.5+k)}{2 \Gamma(k) \sqrt{\frac{c_0 R^{-2 \alpha_d}}{T}-\sigma^2} } \right)
\label{eqn:Ps_kappamufading}
\end{equation}
where  $k$=$\frac{ \mu (1+\kappa)^2}{1+2  \kappa}$ , $\theta=\frac{1}{k}$.
\label{cor:Ps_kappamu_approximate}
\end{cor}
\begin{proof}
$\kappa$-$\mu$ fading is a special case of $\kappa$-$\mu$ shadowed fading when shadowing parameter $m \rightarrow \infty$. Ranging success probability is derived from \eqref{eqn:Ps_kappamushadowed_approximate} for  $m \rightarrow \infty$.
Follows the same steps as the proof in Corollary \ref{cor:kappamu}, to derive ranging success probability.
\end{proof}
\subsection{Rician faded interferer channels}
Ranging success probability is obtained from \eqref{eqn:Ps_kappamufading} by substituting $\mu$=1, $\kappa$=K, where $K$ is the Rician shape parameter. 

Ranging success probability for Nakagami, Rayleigh and no fading cases in Corollary \ref{cor:nakagami_nodesired} - \ref{cor:nofading_nodesired} are exact expressions and can be derived from either \eqref{eqn:Ps_kappamushadowedfading} or \eqref{eqn:Ps_kappamushadowed_approximate}. 

\subsection{Nakagami-m faded interferer channels}
\begin{cor}
Ranging success probability when interferers experience Nakagami-$m$ fading of shape parameter $(\hat{m})$, scale parameter ($\frac{1}{\hat{m}}$) is
\begin{equation}
P_s=  Erfc \left(\frac{\rho_I \sqrt{\pi \gamma_0}   \Gamma(0.5+\hat{m}) }{2 \Gamma(\hat{m}) \sqrt{\hat{m} (\frac{c_0 R^{-2 \alpha_d}}{T} -\sigma^2)}} \right) 
\label{eqn:Ps_nakagami}
\end{equation}
\label{cor:nakagami_nodesired}
\end{cor}
\begin{proof}
Ranging success probability is derived from \eqref{eqn:Ps_kappamufading} for $\mu=\hat{m}$, $\kappa=0$.
\end{proof}
\subsection{Rayleigh faded interferer channels}
\begin{cor}
When interferers experience Rayleigh fading, ranging success probability is  
\begin{equation}
    P_s=  Erfc \left(\frac{\rho_I \sqrt{ \gamma_0}   \pi }{4  \sqrt{ (\frac{c_0 R^{-2 \alpha_d}}{T} -\sigma^2)}} \right)  
    \label{eqn:Ps_rayleigh}
\end{equation}
This is derived from \eqref{eqn:Ps_nakagami} by substituting $\hat{m}=1$.
\end{cor}
\subsection{No fading in interferer channels}
In this Section we derive ranging success probability for no fading as a special case, which is verified with equation (33) derived in \cite{radar}. 
\begin{cor}
When interferers experience no fading, ranging success probability is
\begin{equation}
P_s=  Erfc \left(\frac{\rho_I \sqrt{\pi \gamma_0}    }{2  \sqrt{ (\frac{c_0 R^{-2 \alpha_d}}{T} -\sigma^2)}} \right) 
\label{eqn:Ps_nofading}
\end{equation}
\label{cor:nofading_nodesired}
\end{cor}
\begin{proof}
Interferers experience no fading when $\hat{m} \rightarrow \infty$. 
Using the identity \cite{sudharsanevm}, for large n
\[  \frac{\Gamma(n+a)}{\Gamma(n+b)} = \frac{1+\frac{(a-b)(a+b-1) +O(\frac{1}{n^2})}{2n}}{n^{b-a}},  \]
Using the above identity, as $\hat{m} \rightarrow \infty$, 
\begin{equation}
    \frac{\Gamma(0.5+\hat{m})}{\sqrt{\hat{m}}\Gamma(\hat{m})} \rightarrow 1.
    \label{eqn:gammaidentity}
\end{equation}
Substituting \eqref{eqn:gammaidentity} in \eqref{eqn:Ps_nakagami}, ranging success probability in \eqref{eqn:Ps_nofading} is derived, and is same as equation (33) in \cite{radar}.
\end{proof}
\section{Results}
The plots below are for $\alpha_d$=$\alpha_I$ unless stated otherwise. The parameters used for all these plots unless stated otherwise is unity for distance to target (R), $\gamma_0$, $\gamma_1$, interferer density ($\rho_I$) and path loss exponent is 4. In plots Fig. 6-Fig. 9, path loss exponents used are $\alpha_d$=4, $\alpha_I$=2. In Fig. \ref{fig:nakagami_rayleighdesired}, ranging success probability is plotted for interference limited system  (no noise), when desired channels experience Rayleigh fading and the interferers experience Nakagami/Rayleigh fading. The theoretical results derived in Corollary \ref{cor:rayleighdesired_nakagami}  are verified with simulation. This plot is for number of weights $N$=10 and  shows that it is sufficient for a wide range of mean of the desired channel power ($\overline{g_p}$) from -10 dB to 20 dB. With increase in $\hat{m}$ i.e. as interfering channel becomes stronger with lesser fading, $P_s$ decreases. For rest of plots, unit mean desired channel power  ($\overline{g_p}$) is assumed.

In Fig. \ref{fig:kappamushadowed_rayleighdesired},  ranging success probability is plotted for interference limited system, when desired channels experience Rayleigh fading and the interferers experience $\kappa$-$\mu$ shadowed fading. The plot is for $m$=1 (heavy shadowing) case. For different combinations of $\kappa$, $\mu$ we show that using the approximation of  $\kappa$-$\mu$ shadowed fading with a single Gamma distribution matches well with simulation. The analytical expression in Corollary \ref{cor:kappamushadowed_rayleigh_approximate} is verified in this plot. 

In Fig. \ref{fig:kappamu_rayleighdesired},  ranging success probability is plotted for interference limited system, when desired channels experience Rayleigh fading and the interferers experience $\kappa$-$\mu$ fading. The plot is for $m$=$\infty$ (no shadowing) case. For different combinations of $\kappa$, $\mu$ we show that using the approximation of  $\kappa$-$\mu$ shadowed fading with a single Gamma distribution matches well with simulation.  The analytical expression in Corollary \ref{cor:kappamu} is verified in this plot. For Rician fading ($\mu$=1) as the Rician factor $\kappa$ in interferer increases from 0 to 20, $P_s$ decreases , as the interferer becomes stronger with increase in LOS component $\kappa$. 

In Fig. \ref{fig:impactofkappamum}, the impact of parameters $\kappa$, $\mu$ and $m$ on ranging success probability is plotted when interferers experience $\kappa$-$\mu$ shadowed fading.  In Fig. \ref{fig:kappamu_rayleighdesired} it was shown that when $m>\mu$ i.e. $m=\infty$, $\mu=1$, with increase in $\kappa$, $P_s$ decreases. Similar effect is observed here too when $m>\mu$. When $m<\mu$,  with increase in $\kappa$, we observe increase in $P_s$. And when $m=\mu$, change in $\kappa$ does not affect the ranging success probability. So when there is no shadowing ($m=\infty$), m will be always greater than $\mu$ and with increase in LOS component $\kappa$ in interferer, $P_s$ decreases. Whereas when there is full shadowing  ($m$=0.5), m will be always lesser than number of clusters $\mu$ and with increase in LOS component $\kappa$ in interferer, $P_s$ increases. So the effect of LOS component $\kappa$ on $P_s$ depends on the relationship between shadowing component (m) and number of clusters ($\mu$). To the best of our knowledge, this is the first time such an observation has been made in literature. 

In Fig. \ref{fig:rayleigh_rayleigh_L}, when the longitudinal distance is $\delta_0$=0, the impact of increase in inter-lane distance $L$ is more profound than for $\delta_0=2$. So for planners, if the longitudinal distance is very less, then to increase the ranging success probability, increasing the longitudinal distance will be a good solution. Similarly when the inter-lane distance is $L$=0, increasing the longitudinal distance $\delta_0$ has more profound impact than doing the same when inter-lane distance is $L$=2. So in highly dense cities with narrow roadways, where $L$ is very low, increasing $\delta_0$ will help in a significant increase of ranging success probability. In this plot, expression in Corollary \ref{cor:rayleigh_L} is also verified with simulation. 

In Fig. \ref{fig:may4}, ranging success probability is plotted for different interfering fading channels when desired channel experiences no fading. Ranging success probability when interferer experiences $\kappa$-$\mu$ shadowed fading (Corollary \ref{cor:Ps_kappamushadowed_approximate})  is verified by simulation. Ranging success probability when interferer experiences $\kappa$-$\mu$ fading ($m \rightarrow \infty$ Corollary \ref{cor:Ps_kappamu_approximate} ) is verified by simulation. Similarly ranging success probability (Cor. \ref{cor:nakagami_nodesired}) when interferers experience Nakagami-m fading (m $\rightarrow \infty$, $\kappa$=0, $\mu=\hat{m}$) is verified by simulation.  We also compare the ranging success probability when desired channel experiences Rayleigh fading to desired channel not experiencing fading. There is a significant decrease in ranging success probability when desired channel experiences Rayleigh fading. When desired  channel experiences fading, and we design the system assuming there is no fading, will lead to  aggressive design resulting in  significant errors. This implies that to achieve a certain ranging success probability, we will design the system assuming it will require higher SINR threshold than what will be actually needed. This shows the importance of the expressions derived in this paper for desired channel being Rayleigh faded.

In Fig. \ref{fig:noise}, ranging success probability is plotted for a system affected by both interference and additive white Gaussian noise (0 dB). With increase in $\mu$ (number of clusters), interference becomes stronger thus reducing the ranging success probability. With decrease in $m$ (increasing shadowing), interference weakens thus resulting in increase of ranging success probability. The parameter $\gamma_0$ is present in both the signal and interference component. In an interference limited system, change in $\gamma_0$ does not impact $P_s$. In system with AWGN noise, change in $\gamma_0$ impacts $P_s$. We observe that with increase in parameter $\gamma_0$, ranging success probability increases.  The parameters used are $\alpha$=2, $\gamma_1=2$, $R$=1, $\rho_I=2$.

In Fig. \ref{fig:nofading}, when desired channel experiences no fading,  ranging success probability is plotted when interferers experience  Rayleigh fading \eqref{eqn:Ps_rayleigh} and no-fading \eqref{eqn:Ps_nofading} along with simulation results.  When interferers experience no-fading i.e. when interferers are stronger, $P_s$ is lesser than the case when interferers experience Rayleigh fading. When interferers experience fading, and we design the system assuming there is no fading in interferers, will lead to  conservative design. This implies that to achieve a certain ranging success probability, we will design the system assuming it will require lower SINR threshold than what will be actually needed. This shows the importance of the expressions derived in this paper for interferer being Rayleigh faded \eqref{eqn:Ps_rayleigh}.

In Fig. \ref{fig:may10}, ranging success probability is plotted for a constant path loss exponent of desired signal $\alpha_d$=3.  When longitudinal distance $\delta_0$ is very low (0), and path loss exponent of interferer ($\alpha_I$) increases, then ranging success probability decreases, due to increase in interference power. Interference power increases with increase in $\alpha_I$, as the interferers can be lesser than one unit distance, since $\delta_0$ is very low. When $\delta_0$ is moderately high ($>1$), and path loss exponent of interferer ($\alpha_I$) increases,  then ranging success probability increases due to reduction in interference power.

In Fig. \ref{fig:may10_1}, ranging success probability is plotted for a constant path loss exponent of interference signal $\alpha_I$=3.  When the distance to the desired target (R) is very low (0.5), and path loss exponent of desired signal ($\alpha_d$) increases, then ranging success probability increases  due to increase in desired signal power. Desired signal power increases with increase in $\alpha_d$, as the desired target is lesser than one unit distance. When the distance to the desired target is moderately high (R=1.5), and path loss exponent of desired signal ($\alpha_d$ ) increases,  then ranging success probability decreases due to reduction in desired signal power.   

\begin{figure}
\includegraphics[scale=1]{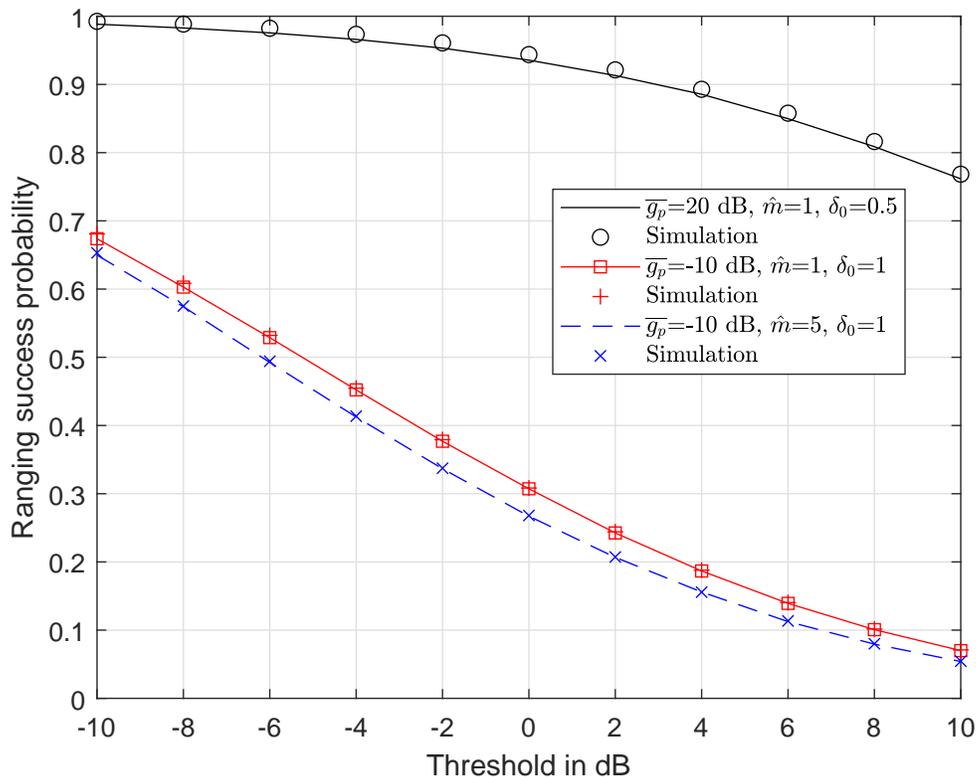}
\centering
\caption{Ranging success probability when interferer experiences Nakagami-m/Rayleigh fading}
\label{fig:nakagami_rayleighdesired}
\end{figure}

\begin{figure}
\includegraphics[scale=1]{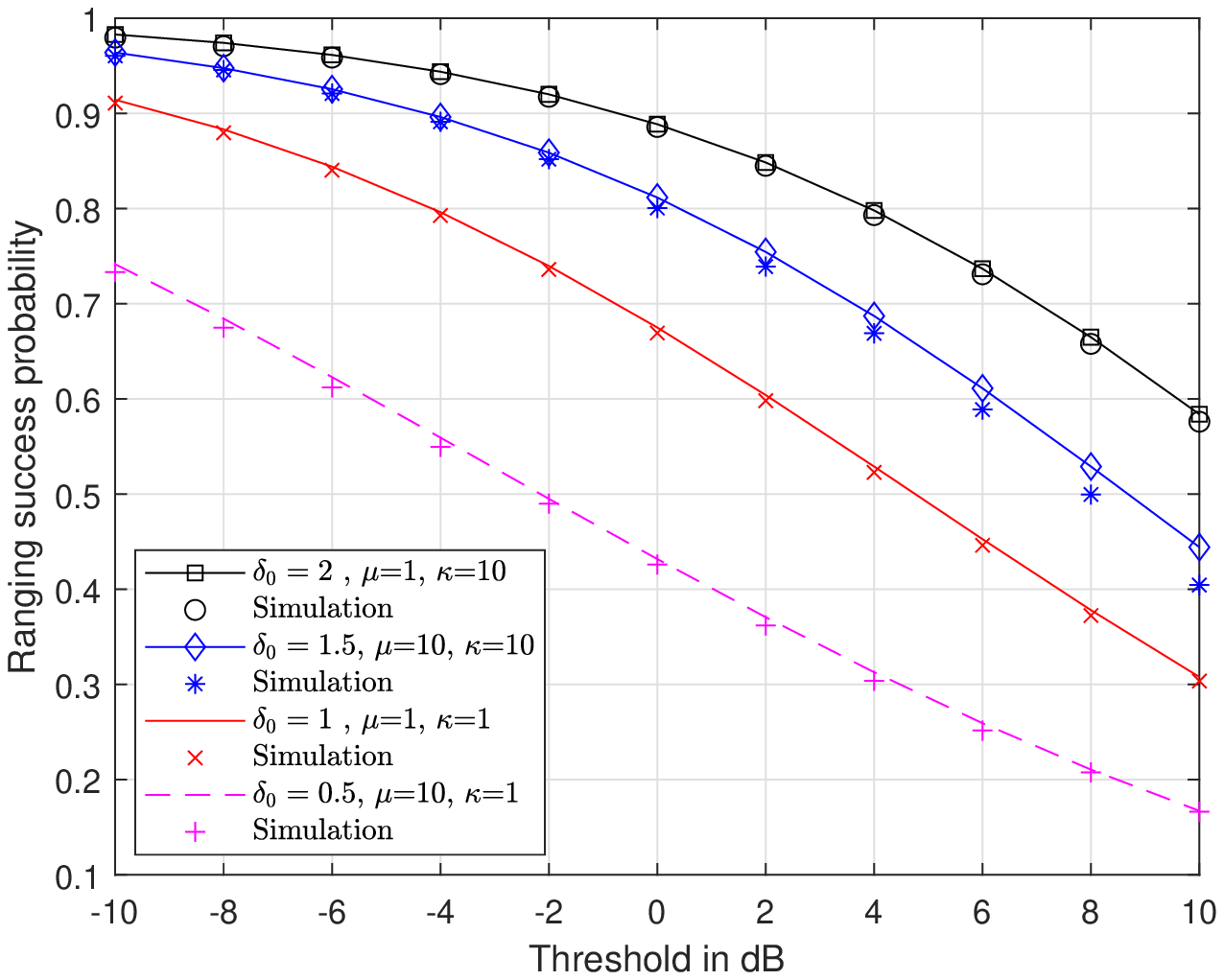}
\centering
\caption{Ranging success probability when interferer experiences $\kappa$-$\mu$ shadowed  fading with m=1}
\label{fig:kappamushadowed_rayleighdesired}
\end{figure}

\begin{figure}
\includegraphics[scale=1]{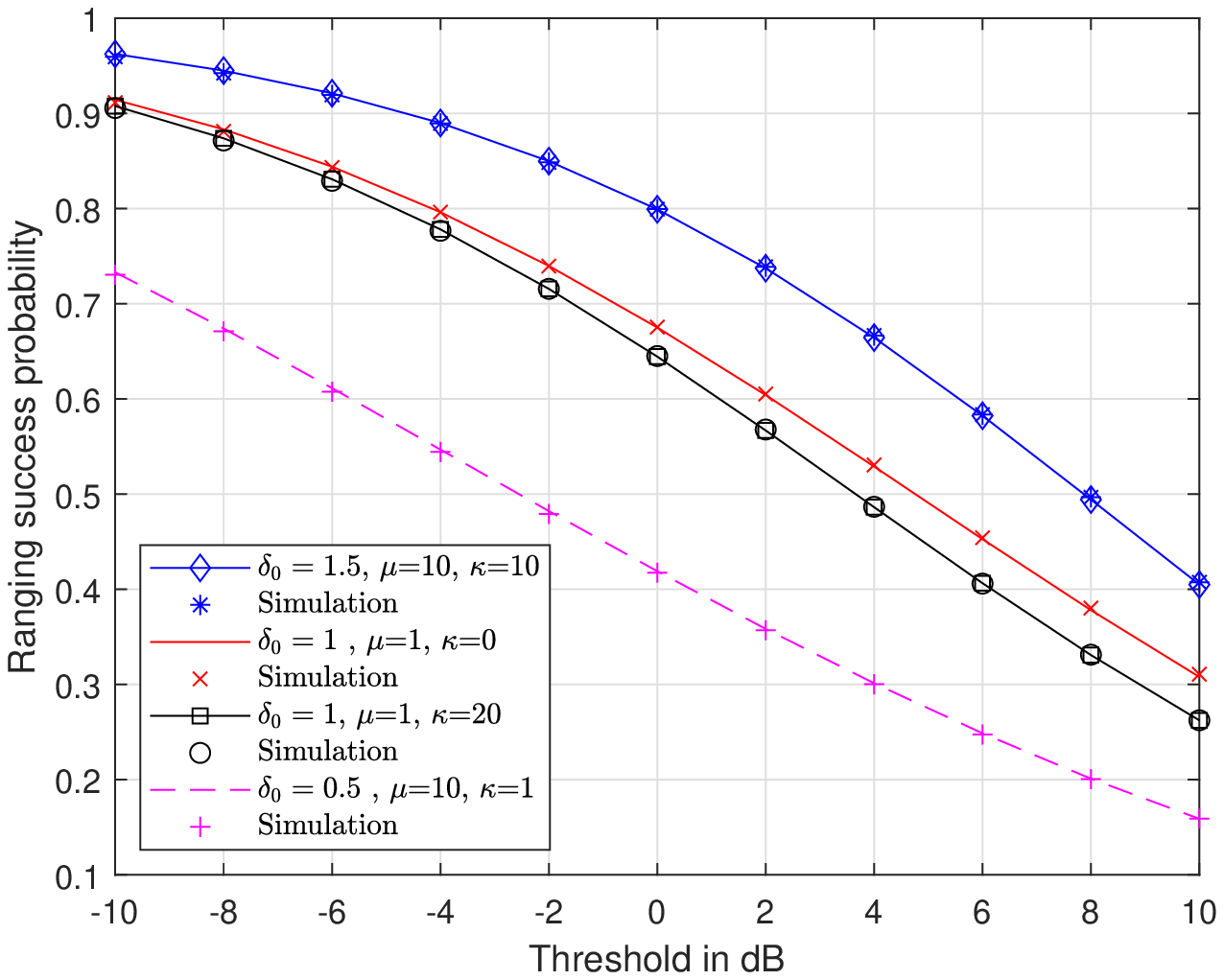}
\centering
\caption{Ranging success probability when interferer experiences $\kappa$-$\mu$  fading}
\label{fig:kappamu_rayleighdesired}
\end{figure}

\begin{figure}
\includegraphics[scale=1]{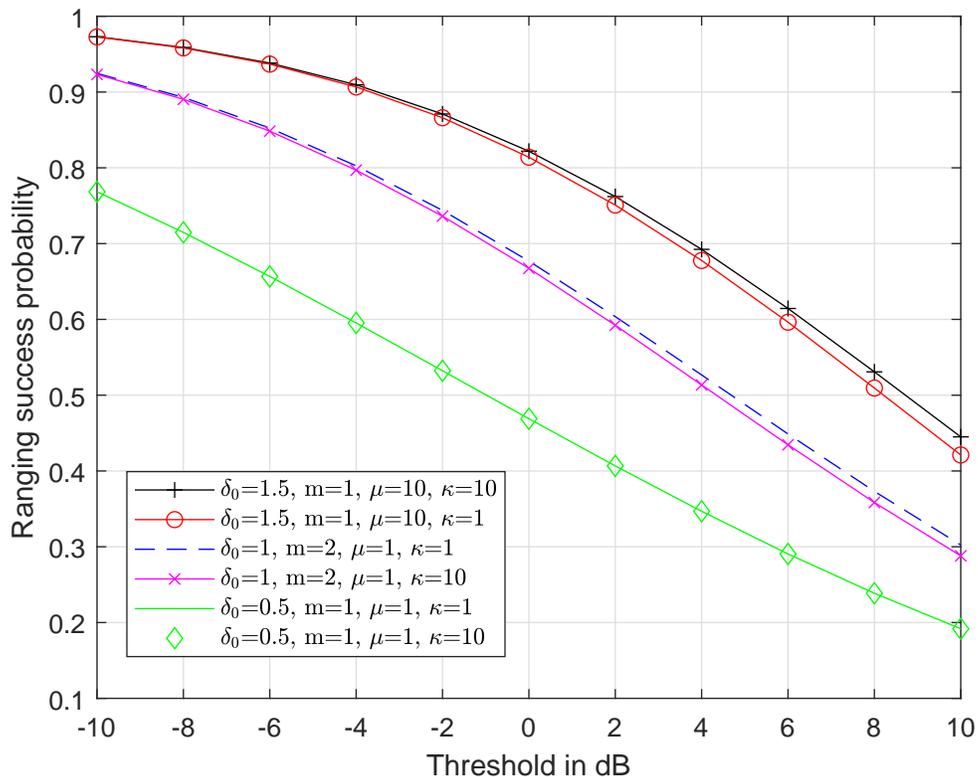}
\centering
\caption{Impact of parameters of $\kappa$-$\mu$ shadowed fading on ranging success probability}
\label{fig:impactofkappamum}
\end{figure}

\begin{figure}
\includegraphics[scale=1]{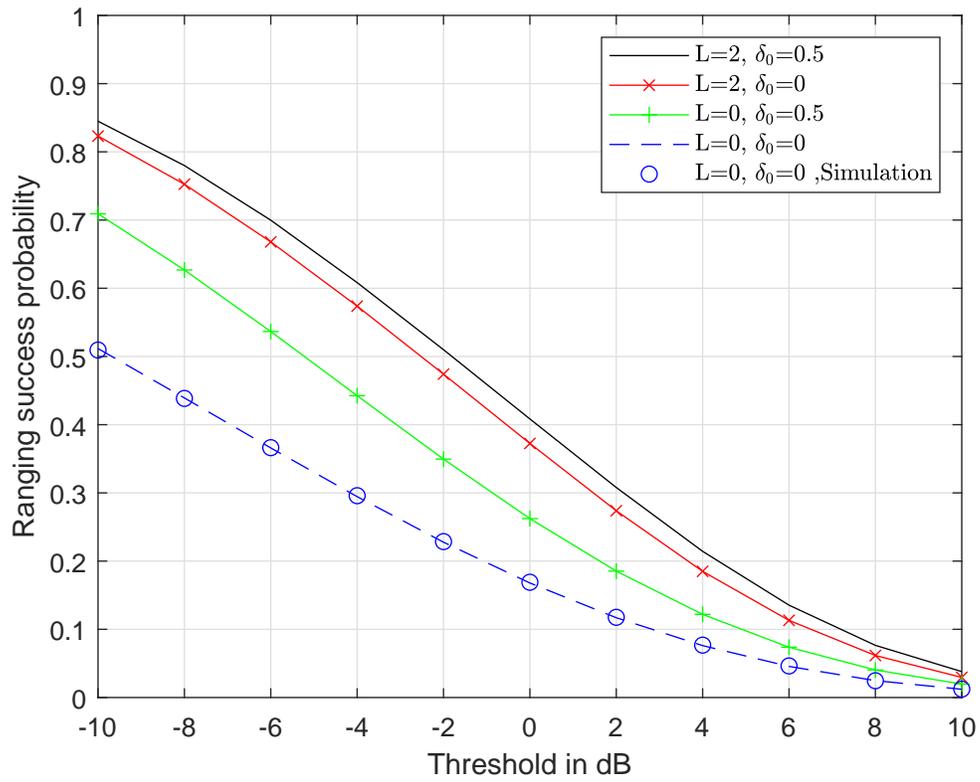}
\centering
\caption{Ranging success probability for arbitrary inter-lane distance when interferer experiences Rayleigh fading}
\label{fig:rayleigh_rayleigh_L}
\end{figure}

\begin{figure}
\includegraphics[scale=1]{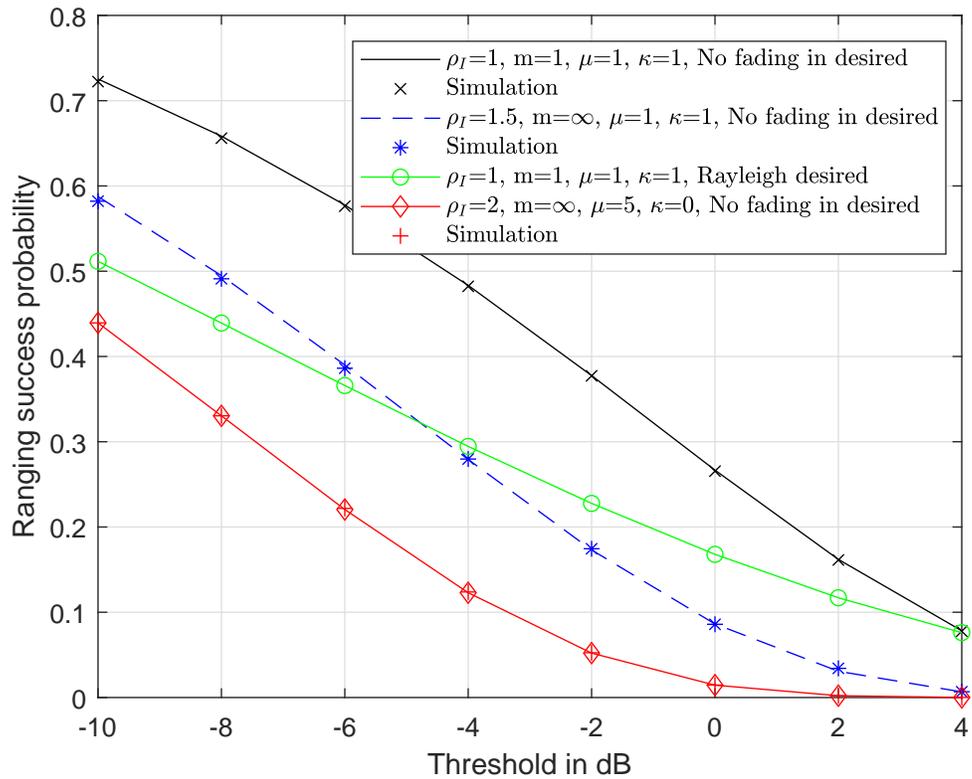}
\centering
\caption{Ranging success probability when desired channel experiences No-fading or Rayleigh fading}
\label{fig:may4}
\end{figure}

\begin{figure}
\includegraphics[scale=1]{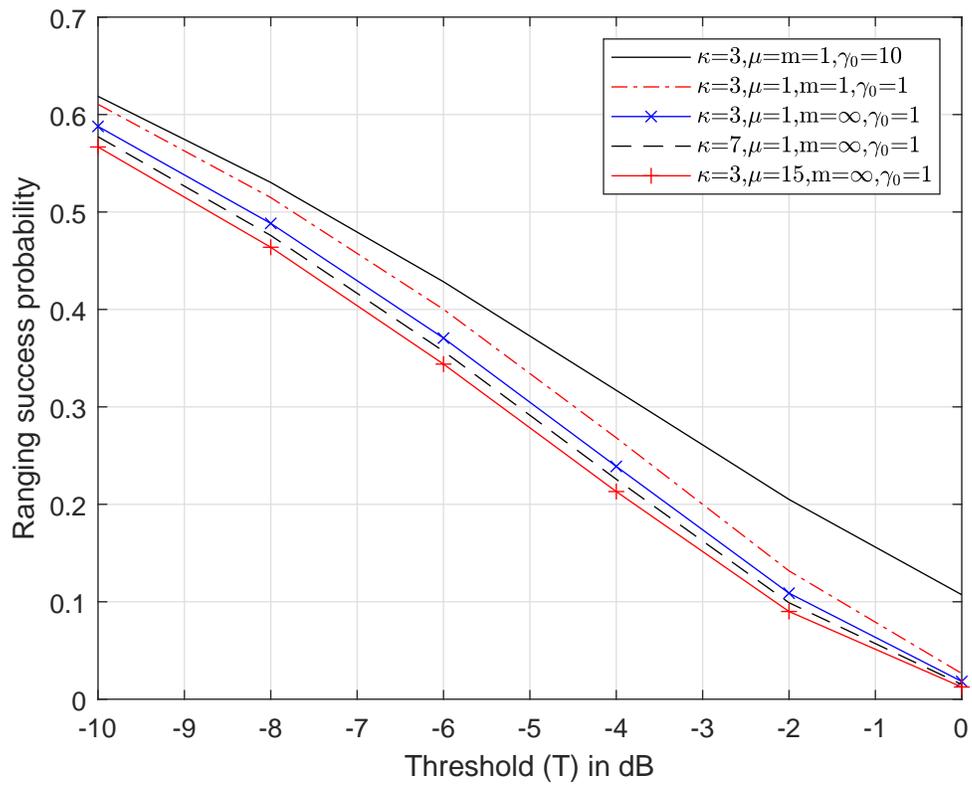}
\centering
\caption{Ranging success probability in presence of noise}
\label{fig:noise}
\end{figure}
\begin{figure}
\includegraphics[scale=1]{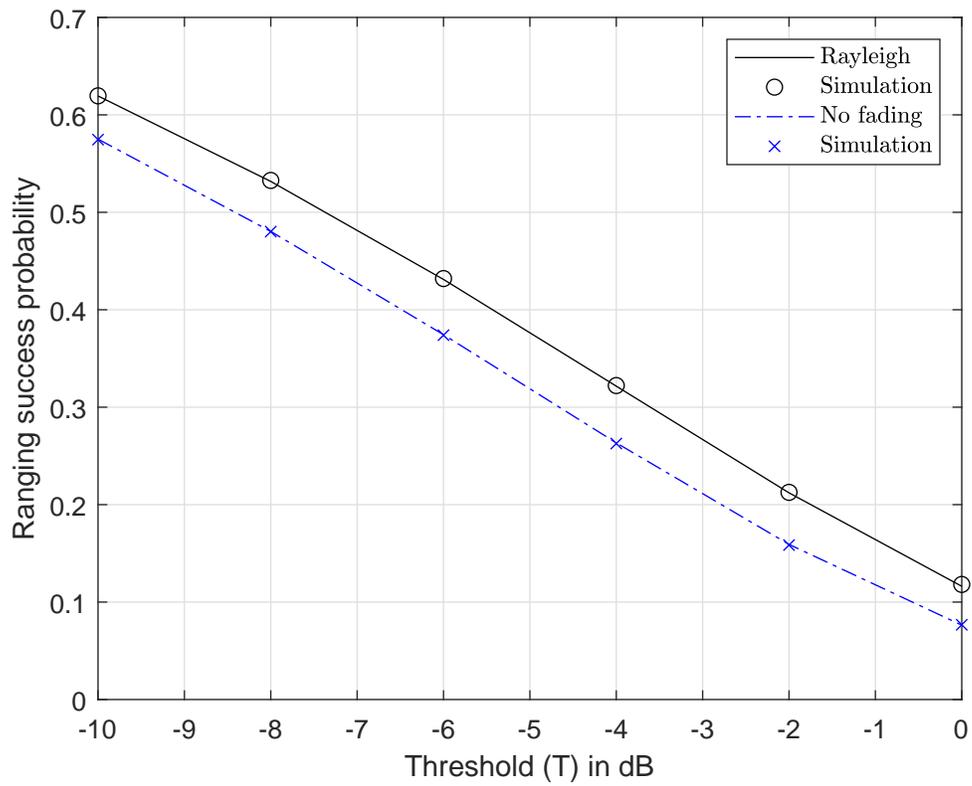}
\centering
\caption{Ranging success probability in presence of no fading and Rayleigh fading in interfering channels}
\label{fig:nofading}
\end{figure}

\begin{figure}
\includegraphics[scale=0.8]{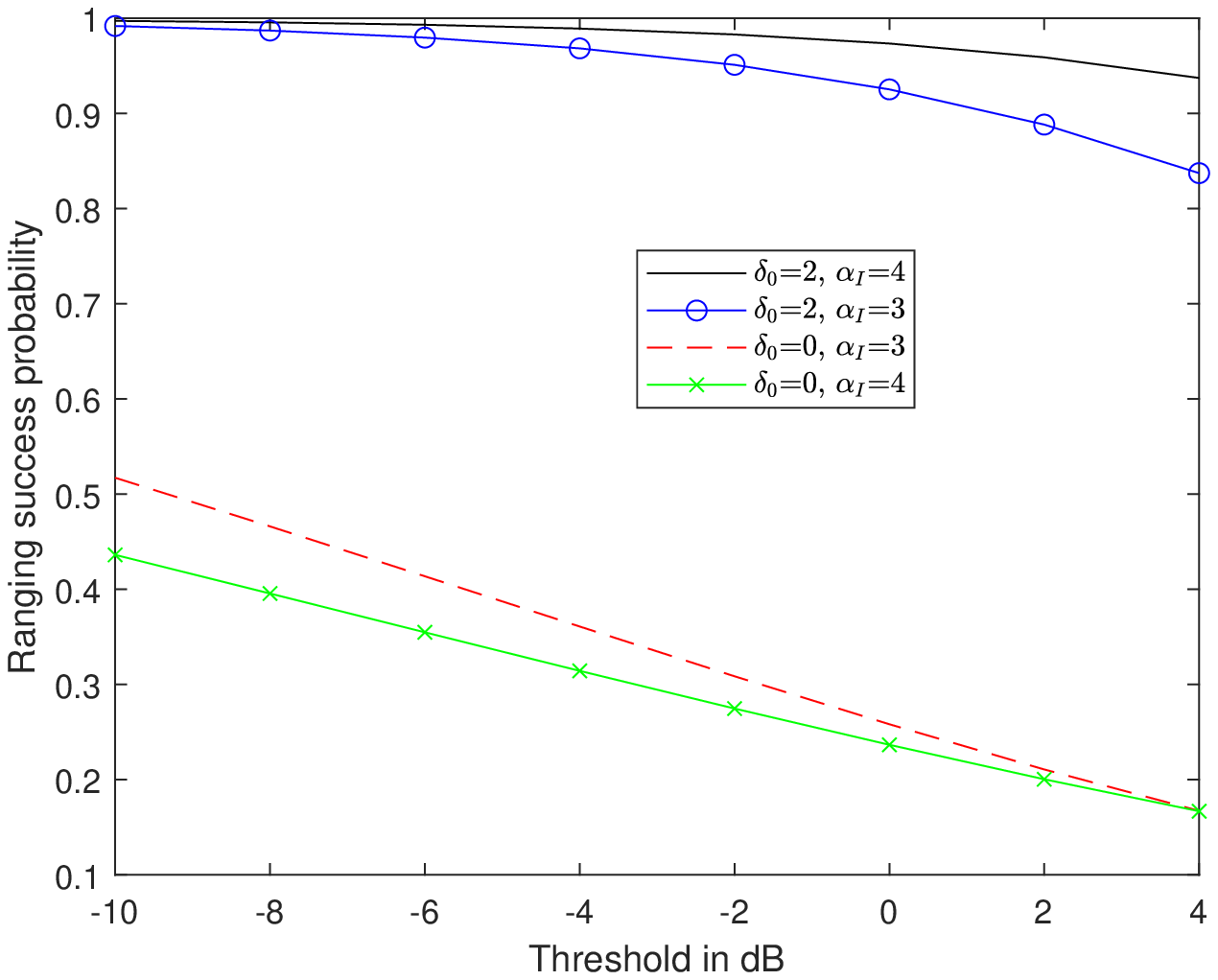}
\centering
\caption{Ranging success probability for constant $\alpha_d$}
\label{fig:may10}
\end{figure}

\begin{figure}
\includegraphics[scale=0.8]{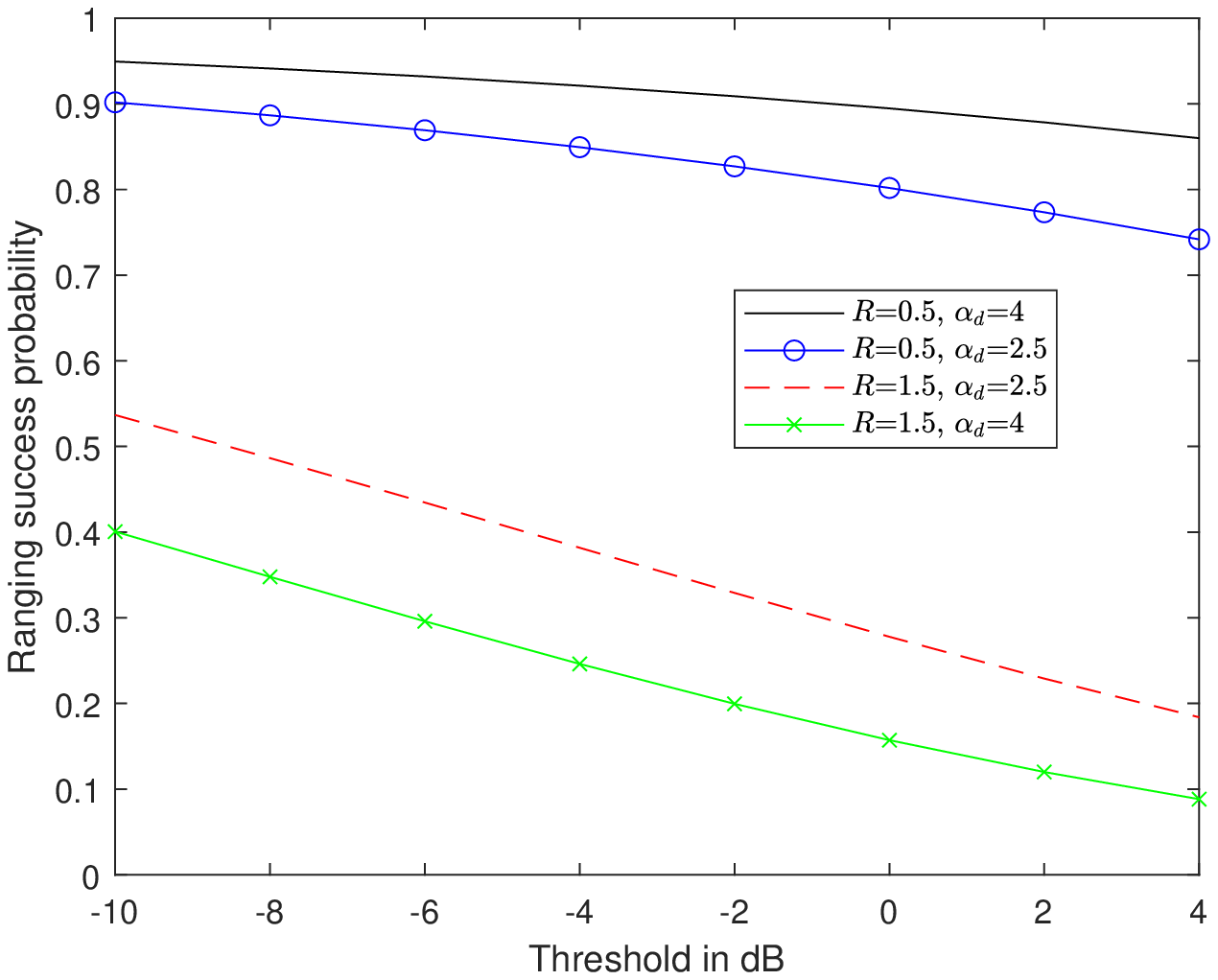}
\centering
\caption{Ranging success probability for constant $\alpha_I$}
\label{fig:may10_1}
\end{figure}

\section{Conclusions and Future Work}
In this paper we have derived the ranging success probability of automotive radar system, considering Poisson point process  modeled interferers. State-of-art literature had derived closed form expressions considering no-fading in desired channel and interferers. We have assumed generalized $\kappa$-$\mu$ shadowed fading in interferers. For desired channel, we have considered both Rayleigh and no-fading cases which required different methods to derive the results. 

As $\kappa$-$\mu$ shadowed fading generalizes  Rayleigh, Rician, Nakagami-m fading etc., ranging success probability for these fading models are also derived. When we consider Rayleigh fading,  in radar applications, the total desired channel will be Rayleigh-Rayleigh. To make the analysis amenable for PPP network assumption, we have expressed the PDF of power of Rayleigh-Rayleigh channel as  weighted sum of exponential PDF. This helped in expressing ranging success probability in terms of Laplace transform of interference. To the best of our knowledge, even the state-of-art PPP literature has not dealt with product ``Rayleigh-Rayleigh" channels.  

Also, to the best of our knowledge, for the first time, in this paper we have studied how the relationship between shadowing component (m) and number of clusters ($\mu$) can affect the impact of LOS component ($\kappa$) on ranging success probability. We have shown that when there is full shadowing ($m$=0.5), with increase in LOS component $\kappa$ in interferer, $P_s$ increases. Similarly when there is no shadowing component ($m$=$\infty$), with increase in LOS component $\kappa$ in interferer, $P_s$ decreases.

We have also made recommendations for when it would be best for city planners to increase longitudinal distance ($\delta_0$) for a given inter-lane distance or vice versa.  We have also observed how the value of distance to desired target ($R$) can influence the impact of change in desired path loss exponent ($\alpha_d$) on ranging success probability. Similarly, we have also observed how the value of longitudinal distance ($\delta_0$) can influence the impact of change in interferer path loss exponent ($\alpha_I$) on ranging success probability. 

We have also shown the importance of results derived in this paper assuming fading, comparing it with no-fading cases in desired and interfering channels. By not considering fading in desired channels,  to achieve a certain ranging success probability, we will design the system assuming it will require higher SINR threshold than what will be actually needed. By not considering fading in interferer channels,  to achieve a certain ranging success probability, we will design the system assuming it will require lower SINR threshold than what will be actually needed. The formulae derived for different fading scenarios will help system planners get an idea of SINR threshold required to achieve target success probability, for any of the popular fading that they observe in real life scenarios.

The results derived in this paper for Rayleigh fading in desired channels can be easily extended to Nakagami-m fading channels. But to consider generalized $\kappa$-$\mu$ shadowed fading in desired channels and derive simple closed form results is a significant open problem.
\bibliographystyle{IEEEtran}
\bibliography{Bib}
\end{document}